\documentclass[11pt,twoside]{article}
\usepackage{url}
\usepackage{hyperref}
\usepackage{jeffe,graphicx}
\usepackage{euscript}
\usepackage[utf8]{inputenc}
\usepackage[margin=1in]{geometry}
\usepackage{marvosym}
\usepackage{cite}
\usepackage{verbatim}
\usepackage{color}
\usepackage{enumerate}
\usepackage{tikz}
\usepackage{array}
\usepackage{cite}

\usepackage{microtype}
\usepackage{stmaryrd}
\usepackage{comment}

\makeatletter
 \def\section{\@startsection {section}{1}{\z@}%
                                      {-1.75ex \@plus -1ex \@minus -.2ex}%
                                      {1.0ex \@plus.2ex}%
                                      {\normalfont\Large\bfseries}}
 \renewcommand\subsection{\@startsection{subsection}{2}{\z@}%
                                        {-1.5ex\@plus -1ex \@minus -.2ex}%
                                        {0.75ex \@plus .2ex}%
                                        {\normalfont\large\bfseries}}
 \renewcommand\paragraph{\@startsection{paragraph}{4}{\z@}%
                                       {1.0ex \@plus1ex \@minus.2ex}%
                                       {-1em}%
                                       {\normalfont\normalsize\bfseries}}

 \makeatother
\makeatletter
\long\def\@makecaption#1#2{
   \vskip 10pt 
   \setbox\@tempboxa\hbox{{\footnotesize \textbf{#1.} #2}} 
   \ifdim \wd\@tempboxa >\hsize         
       {\footnotesize \textbf{#1.} #2\par}
     \else                              
       \hbox to\hsize{\hfil\box\@tempboxa\hfil}
   \fi} 
\makeatother

\def\compilelong{}
\ifx\compilelong\undefined
  \excludecomment{longversion}
  
\else
  \excludecomment{shortversion}
  
\fi

\def\Aset{A}

\def\ball{\mathsf{B}}

\def\child{\mathsf{ch}}
\def\d{\mathsf{d}}  

\def\dsq{\mathsf{d}_{\sq}}  

\def\dfr{\mathsf{dfr}}

\def\dtw{\mathsf{dtw}}
\def\dmin{\underline{\d}}
\def\dmax{\overline{\d}}

\def\ed{\mathsf{ed}}
\def\eps{\varepsilon}

\def\family{\EuScript{F}}
\def\frechet{Fr\'{e}chet~}

\def\H{H}

\def\pair{\textsc{Pairing}}
\def\parent{p}

\def\rmin{\underline{r}}
\def\rmax{\bar{r}}

\def\sq{\square}

\def\T{\EuScript{T}}

\def\tQ{\tilde{Q}}

\def\V{\mathsf{V}}
\def\rect{\mathsf{R}}
\def\rects{\EuScript{R}}
\def\boundaries{\EuScript{B}}
\def\wei{\omega}

\newcommand*\samethanks[1][\value{footnote}]{\footnotemark[#1]}

\newtheorem{theorem}{Theorem}[section]
\newtheorem{corollary}[theorem]{Corollary}
\newtheorem{lemma}[theorem]{Lemma}

\begin{document}

\title{Approximating Dynamic Time Warping and Edit Distance for a Pair of Point Sequences
  \footnote{Work on this paper is supported by NSF under grants CCF-11-61359,
  IIS-14-08846, and CCF-15-13816, by an ARO grant W911NF-15-1-0408, and by Grant
2012/229 from the U.S.-Israel Binational Science Foundation.}}

\author{
  Pankaj K. Agarwal%
  \thanks{Department of Computer Science,
  Duke University; \url{[pankaj,kylefox,jwpan]@cs.duke.edu}, \url{zhitao.ying@duke.edu}.}
  \and
  Kyle Fox\samethanks
  \and
  Jiangwei Pan\samethanks
  \and
Rex Ying\samethanks}

\maketitle
\begin{abstract}
  We give the first subquadratic-time approximation schemes for dynamic time warping (DTW) and
  edit distance (ED) of several natural families of point sequences in~$\R^d$, for any fixed
  $d \ge 1$.
  In particular, our algorithms compute $(1+\eps)$-approximations of DTW and ED in time
  near-linear for point sequences drawn from k-packed or k-bounded curves, and
  subquadratic for backbone sequences.  Roughly speaking, a curve is
  $\kappa$-packed if the length of its intersection with any ball of radius $r$ is at
  most $\kappa \cdot r$, and a curve is $\kappa$-bounded if the sub-curve between two curve points
  does not go too far from the two points compared to the distance between the
  two points.  In backbone sequences, consecutive points are spaced at
  approximately equal distances apart, and no two points lie very close
  together.
  Recent results suggest that a subquadratic algorithm for DTW or ED is unlikely for an
  arbitrary pair of point sequences even for $d=1$.

  Our algorithms work by constructing a small set of rectangular regions that cover the
  entries of the dynamic programming table commonly used for these distance
  measures.
  The weights of entries inside each rectangle are roughly the same, so we are able to use
  efficient procedures to approximately compute the cheapest paths through these rectangles.
\end{abstract}

\section{Introduction}
\label{sec:intro}

\paragraph*{Motivation.}
Trajectories are functions from a time interval to $\R^d$, for $d \ge 1$, and they
describe how physical systems change over time.  Trajectories are being recorded and
inferred from numerous sources and are often represented as ordered sequences of points.
These sources include GPS sensors in smart phones and vehicles, surveillance videos,
shape-based touch screen authentication patterns, hurricane patterns, and time series
data.  A fundamental task for analyzing trajectory data is to measure the similarity
between trajectories.  For example, computing trajectory similarity is an important step in
object segmentation from video trajectories \cite{BM10}, smart phone authentication using
touch screen trajectories \cite{DHBLH12}, and stock price prediction using historical
patterns \cite{STOCK}.  In many applications, it is not enough to merely \emph{quantify}
how similar pairs of trajectories are; we need to compute correspondences between their
sample points as well.  These correspondences represent shared structures between
trajectories, which can be present not only in trajectories with physical constraint such
as vehicle trajectories, but also in trajectories representing the movement of
pedestrians~\cite{ghb-uihmp-08} or hurricanes~\cite{lhw-tcpgf-07}.  Having an effective
way to identify similar portions between a pair of trajectories can greatly aid in
identifying and understanding these shared structures.  

\paragraph*{Problem statement.}
Let~$P = \langle p_1, \ldots, p_m \rangle$ and~$Q = \langle q_1, \ldots, q_n \rangle$ be
two sequences of points in~$\R^d$ for some fixed $d \ge 1$.
We define a \emph{correspondence} as a pair~$(p_i,q_j)$.
A set $C$ of correspondences is \emph{monotone} if for any pair of correspondences
$(p_i,q_j), (p_{i'},q_{j'})$ with $i' \geq i$ we have $j' \geq j$.
We define the cost of $C$ to be $\sum_{(p,q) \in C} ||pq||$, where $||\cdot||$
is the Euclidean norm.  The similar portions of $P$ and $Q$ are represented by a
set $C$ of monotone correspondences, with the cost of $C$ quantifying the quality of
similarity.
The goal is to compute a monotone set of correspondences with certain properties.
While numerous criteria for computing correspondences have been proposed, we
focus on two, which are widely used: dynamic time warping (DTW) and edit
distance (ED).  They are used for matching various types of sequences such as 
speech signals, DNA and protein
sequences, protein backbones, time-series data, GPS/video trajectories, touch
screen authentication trajectories, etc \cite{rabiner1993fundamentals, WG, EFS, DEKM-book,
  DHBLH12, KKL, MP99, KP00, KR05, Muller-book, CPBTMPCMC-98}.

DTW computes a monotone set of correspondences in which every point in~$P$
and~$Q$ appears at
least once, and minimizes the sum of distances of corresponding pairs of points.
Formally, the cost of DTW, denoted by $\dtw(P,Q)$, is 
  $\dtw(P,Q) = \min_{C} \sum_{(p,q) \in C} ||pq||$,
where the minimum is taken over all sets $C$ of monotone correspondences that
cover all points of $P$ and $Q$.
DTW allows a point to appear in multiple correspondences, so it matches two
sequences effectively
even if the sampling rates are different.

Edit distance (also called Levenshtein distance) seeks a monotone
\emph{matching} on the points in~$P$ and~$Q$ of minimum cost; each point in
$P$ corresponds to at most one point in $Q$ and vice versa.
It also adds an additional constant \emph{gap} penalty~$g$ for each point in~$P
\cup Q$ that
does not appear in any correspondence.
Formally, the cost of ED, denoted by $\ed(P,Q)$, is 
  $\ed(P,Q) = \min_C \sum_{(p,q) \in C} ||pq|| + g(m+n-2|C|)$,
where the minimum is taken over all sets $C$ of monotone 
matchings in the complete bipartite graph $P \times Q$.  More sophisticated
gap penalty functions have been proposed, but we focus on the simple linear gap
penalty function.
By tuning~$g$ correctly, meaningful correspondences can be computed even when faced with
outlier points that arise from measurement errors or short deviations in otherwise similar
trajectories.

Given a parameter $\eps \in (0,1)$, we wish to develop efficient
$(1+\eps)$-approximation algorithms for computing $\dtw(P,Q)$ and $\ed(P,Q)$,
i.e., they return a value $\Delta$ such that $\dtw(P,Q) \le \Delta \le (1+\eps)
\dtw(P,Q)$ or $\ed(P,Q) \le \Delta \le (1+\eps) \ed(P,Q)$.  We are also
interested in computing correspondences that realize these distances.

\paragraph*{Prior results.}
It is well-known that both $\dtw(P,Q)$ and $\ed(P,Q)$, as well as the relevant
correspondences, can be computed in $O(mn)$ time using dynamic programming
\cite{KT2005}.
A series of recent papers show that 
there exists no algorithm for computing $\dtw(P,Q)$ or $\ed(P,Q)$ in
time~$O(n^{2-\delta})$ for any~$\delta > 0$ unless the Strong Exponential Time Hypothesis (SETH)
of Impagliazzo and Paturi~\cite{Impagliazzo2001} is false.
In particular, Backurs and Indyk~\cite{BI2015} showed a lower bound for edit distance, and
Abboud~\etal~\cite{ABW2015} and Bringmann and K\"unnemann~\cite{BK2015}
independently showed lower bounds for DTW.
While most of these lower bounds were presented for the string versions of their respective
problems, the DTW lower bound of Bringmann and K\"unnemann uses sequences of points in~$\R$.
Unless SETH is false, there exists no strictly subquadratic time algorithm for DTW, even in our
setting of point sequences in~$\R^d$.
Similar conditional lower bounds have been shown for other distance and similarity
problems~\cite{AWW2014,BI2015,ABW2015,Bringmann14SETH,BK2015}.
Some of these results suggest that even strongly subquadratic approximation
schemes seem unlikely \cite{ABW2015,Bringmann14SETH}.

In view of the recent lower bounds, a natural question to ask is whether
near-linear, or even subquadratic, algorithms exist for certain natural
families of sequences.  Aronov
\etal~\cite{AHKWW2006} gave subquadratic-time approximation schemes for
the discrete \frechet distance of \emph{$\kappa$-bounded} and \emph{backbone} 
point sequences.
Discrete \frechet distance is similar to DTW except that one uses max instead of
sum in the definition.
Restricting themselves to these families of sequences allowed them to subvert the hardness result of
Bringmann~\cite{Bringmann14SETH} mentioned above.
Driemel~\etal~\cite{DHW2012} extended the approximation results to the continuous
\frechet distance and the even more natural family of \emph{$\kappa$-packed} curves (see
below for definitions of these families of curves);
see also Bringmann and
K\"unnemann~\cite{BK2014} and Bringmann and Mulzer~\cite{BM2015SoCG}.
We note that while (discrete) \frechet distance is a reasonable measure to
compute the similarity between two sequences, it is not effective in identifying
similar portions of the sequences.
Roughly speaking, these algorithms guess the value of (discrete) \frechet distance, say,
$\Delta$, simplify the two sequences within $\frac{\eps}{2}\Delta$ error, and show that
near-linear number of entries in the dynamic-programming table need to be computed by
observing that each point $p$ in one sequence can be matched with a point in the
other (simplified) sequence that lies within distance $\Delta$ from $p$.

Currently, no subquadratic-time approximation results are known for DTW, although there are a number of heuristics
designed to speed up its exact computation in practice; see Wang \etal~\cite{WMDTSK13}.
Subquadratic-time approximation algorithms are known for variants of edit distance, but all of these algorithm have
at least a polylogarithmic approximation ratio~\cite{AKO2010}.
It is not clear how to extend the above algorithms for (discrete) \frechet distance to DTW
or ED because a point $p$ in $P$ may be matched to a point in $Q$ that is far away from
$p$ and thus their idea of computing a small number of entries of the table does not work.

\paragraph*{Our results.}
We present algorithms for computing $\dtw(P,Q)$ and $\ed(P,Q)$ approximately which have
subquadratic running time for several ``well-behaved'' families of input sequences.
The correspondences realizing these distances can also be recovered.  
The two algorithms are almost identical except a few implementation details.
In the worst case, their running time is quadratic for arbitrary point
sequences, but it is near-linear if $P$ and $Q$ are $\kappa$-packed or $\kappa$-bounded
sequences and subquadratic when $P$ and $Q$ satisfy the conditions for a
backbone sequence.  These are the first approximation algorithms that compute
$\dtw(P,Q)$ and $\ed(P,Q)$ for such point sequences in subquadratic time. 

For $x \in \R^d$ and $r \in \R^+$, let $\ball(x,
r)$ denote the ball of radius $r$ centered at $x$.  Given $\kappa \in \R^+$, a
curve $\gamma$ in $\R^d$ is \emph{$\kappa$-packed} if the length of $\gamma$
inside any ball of radius $r$ is bounded by $\kappa r$ \cite{DHW2012}, and
$\gamma$ is \emph{$\kappa$-bounded} if for any $0 \le t < t' \le 1$, 
$ \gamma[t:t'] \subseteq \ball\left(\gamma(t), \tfrac{\kappa}{2} ||\gamma(t)
\gamma(t')||\right) \cup \ball\left(\gamma(t'), \tfrac{\kappa}{2} ||\gamma(t)
\gamma(t')||\right),$
where $\gamma: [0,1] \rightarrow \R^d$ and $\gamma[t:t']$ is the portion of
$\gamma$ between $\gamma(t)$ and $\gamma(t')$ \cite{AKW2004}. 
We say a point sequence $P$ is $\kappa$-packed (resp.~$\kappa$-bounded) if the polygonal
curve that connects points of $P$ is $\kappa$-packed (resp.~$\kappa$-bounded).
A point sequence $P = \langle p_1, \ldots, p_m \rangle$ is said to be a backbone
sequence if it satisfies the following two conditions:
(i) for any pair of non-consecutive integers $i, j \in [1, m]$, $||p_i p_j|| > 1$; 
(ii) for any integer $i$ in $(1, m]$, $c_1 \le ||p_{i-1} p_i|| \le c_2$, 
where $c_1, c_2$ are positive constants \cite{AHKWW2006}.
These sequences are commonly used to model protein backbones where each vertex represents a 
$C_\alpha$ atom, connected to its neighbors via covalent bonds.
See Figure~\ref{fig:alg}(a), (b) for examples of $\kappa$-packed and $\kappa$-packed
curves and (c) for an example of backbone sequence.  
We use $\gamma_P$ to denote the polygonal curve connecting the points of
sequence $P$.
Our results are summarized in the following theorems.


\begin{theorem}
  Let $P$ and $Q$ be two point sequences of length at most $n$ in $\R^d$, and
  let $\eps \in (0,1)$ be a parameter.  An
  $(1+\eps)$-approximate value of $\dtw(P,Q)$ can be computed in
  $O(\frac{\kappa}{\eps} n\log n)$, $O(\frac{\kappa^d}{\eps^d} n \log n)$, and
  $O(\frac{1}{\eps} n^{2-1/d} \log n)$ time if $P,Q$ are $\kappa$-packed,
  $\kappa$-bounded, and backbone sequences, respectively.  
  \label{thm:dtw}
\end{theorem}

\begin{theorem}
  Let $P$ and $Q$ be two point sequences of length at most $n$ in $\R^d$, and
  let $\eps \in (0,1)$ be a parameter.  An
  $(1+\eps)$-approximate value of $\ed(P,Q)$ can be computed in
  $O(\frac{\kappa}{\eps} n\log n)$, $O(\frac{\kappa^d}{\eps^d} n \log n)$, and
  $O(\frac{1}{\eps} n^{2-1/(d+1)} \log n)$ time if $P,Q$ are $\kappa$-packed,
  $\kappa$-bounded, and backbone sequences, respectively.  
  \label{thm:ed}
\end{theorem}

Recall that the standard dynamic programming algorithm for computing $\dtw(P,Q)$
or $\ed(P,Q)$ constructs a weighted grid $\V = \{(i,j) \mid 1\le i
\le m, 1\le j \le n\}$ and formulates the problem as computing a minimum-weight
path from $(1,1)$ to $(m,n)$.  Based on the observation that nearby grid points
may have similar weights when $P,Q$ are ``well-behaved'', our main idea is to construct a small
number of potentially overlapping rectangular regions of $\V$, whose union contains the minimum-weight path in $\V$,
such that all grid points within each rectangle have similar weights.
We show how to construct the rectangles 
so that the number of ``boundary points'' of the rectangles is near linear when $P,Q$ are 
$\kappa$-packed or $\kappa$-bounded and subquadratic when $P, Q$ are 
backbone sequences.
We then describe an efficient procedure to compute approximate minimum-weight paths from $(1, 1)$ 
to all boundary points.

The algorithm framework is quite general and can work for
a variety of similar distance measures based on monotone correspondences.
For example, our results immediately generalize to variants of dynamic time warping and
edit distance that use the $k$-th powers of distance between points instead of their
Euclidean distance for any constant $k > 0$.  Moreover, the framework may prove
useful in designing subquadratic-time algorithms for other problems that can be
solved with standard dynamic programming.

\begin{figure}
  \centering
  \begin{tabular}{ccc}
    \includegraphics[height=2cm]{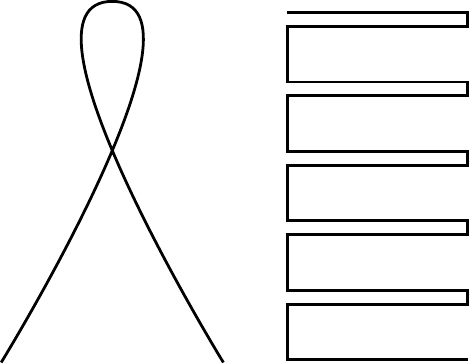} &
    \includegraphics[height=2cm]{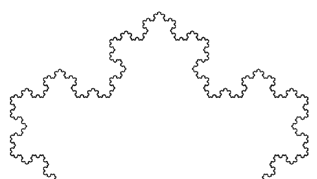} &
    \includegraphics[height=2cm]{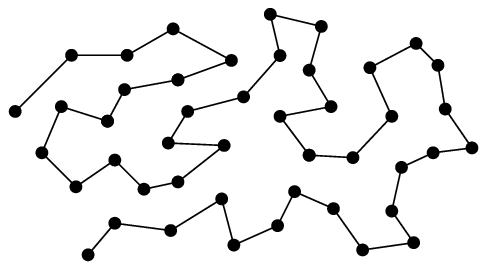} \\
    (a) & (b) & (c)
  \end{tabular}
  \caption{(a) $\kappa$-packed curves that are not $\kappa$-bounded.
    (b) The top half of the Koch snowflake is a $\kappa$-bounded curve that is not $\kappa$-packed.
    (c) A backbone sequence that is neither $\kappa$-bounded nor $\kappa$-packed.}
  \label{fig:alg}
\end{figure}

\def\hsq{\hat{\sq}}
\def\adtw{\hat{\dtw}}
\def\wdtw{\mu} 
\def\awdtw{\tilde{\wdtw}}
\def\Sub{S}
\def\hm{\hat{m}}
\def\hn{\hat{n}}
\def\hL{\hat{L}}

\section{Algorithm for DTW} \label{sec:alg}

Let $P = \langle p_1, \ldots, p_m \rangle$ and $Q = \langle q_1, \ldots, q_n \rangle$ be
two point sequences in $\R^d$, and let $\eps \in (0,1)$ be a parameter.  We present a
$(1+\eps)$-approximation algorithm for computing 
$\dtw(P,Q)$, which takes subquadratic time if $P,Q$ are ``well-behaved''.  
Without loss of generality, assume that $m \le n$ and $\eps \ge
1/n$.  If $\eps < 1/n$, we can simply compute $\dtw(P,Q)$ in $O(mn) = O(n/\eps)$ time via
dynamic programming.

Given positive integers $i < i'$, let $[i:i'] := \{i, i+1, \ldots, i'\}$, and let $[i] :=
[1:i]$.  Let $\V= [m] \times [n] $ denote a set of grid 
points\footnote{Note that in this paper, a point can refer to either a grid point in $\V$ or a
sequence point from $P\cup Q$.}
in $\R^2$, and define a
weight function $\wei: \V \rightarrow \R_{\ge 0}$ where $\wei(i,j)$ is the Euclidean
distance between $p_i$ and $q_j$.  Two different grid points in $\V$ are said to be
\emph{neighboring} if they differ by at most 1 in each coordinate.  We say $(i,j)$
dominates $(i', j')$ if $i \ge i'$ and $j \ge j'$.  A path $\pi$ is a sequence of
neighboring grid points; $\pi$ is \emph{admissible}
if it is non-decreasing in both coordinates.  Define the weight of the path $\pi$,
$\wei(\pi)$, as the sum of the weights of the grid points along $\pi$.  Define $\wdtw(i,j)$
as the minimum weight of an admissible path from $(1,1)$ to $(i,j)$.  So $\dtw(P,Q) =
\wdtw(m,n)$ and $\wdtw(m,n)$ can be computed in $O(n^2)$ time via dynamic
programming.  

For $1 \le i_1 \le i_2 \le m$ and for $1 \le j_1 \le j_2 \le n$, the set of grid points
$[i_1:i_2] \times [j_1:j_2]$ is called a \emph{rectangle}.  A point $(i,j) \in \V$ is a
\emph{boundary point} of this rectangle if $i \in \{i_1, i_2\}$ or $j \in \{j_1, j_2\}$.
We first outline the algorithm for computing an $(1+\eps)$-approximate value of
$\wdtw(m,n)$, and then describe it in detail in Sections~\ref{sec:alg_step1}-\ref{sec:fill_table}.  
Section~\ref{sec:run_time} analyzes its running time for well-behaved
point sequences.

\begin{enumerate}[(i)]
  \item Compute an estimate $\dmin$ of~$\dtw(P,Q)$ such that $\dmin \le \dtw(P, Q) \le 4n
    \dmin$.  Let $\dmax = 4n \dmin$.

  \item Compute a set $\rects$ of (possibly overlapping) rectangles and a weight $\wei_{\rect}$ for each
    rectangle $\rect \in \rects$, such that: 
    \begin{enumerate}[(a)]
      \item for all $\rect \in \rects$ and $(i,j) \in \rect$, $|\wei(i,j) - \wei_{\rect}|
        \le \frac{\eps}{2} \max\{ \wei_{\rect}, \dmin/2n\}$,
      \item if $(i,j) \in \V$ and $\wei(i,j) \le \dmax$, then there exists rectangle
        $\rect \in \rects$ such that $(i,j) \in \rect$. 
    \end{enumerate}

    Conditions (a) and (b) above ensure that the weights of
    grid points in each rectangle are roughly the same, and the minimum-weight admissible
    path between $(1,1)$ and $(m,n)$ is contained in the union of the rectangles.  See
    Figure~\ref{fig:rectangles}(a).  

    Let $\boundaries = \bigcup_{\rect \in \rects} \partial \rect$ be the set of boundary
    points of the rectangles in $\rects$.
    The sizes of~$\rects$ and~$\boundaries$ depend on the input sequences $P, Q$. 
    In the worst case~$|\rects|,|\boundaries| = \Theta(mn)$, but they are subquadratic
    if $P,Q$ are well-behaved.

  \item For every $(i,j) \in \boundaries$, compute a $(1+\eps)$-approximate value
    $\awdtw(i,j)$ of $\wdtw(i,j)$, i.e., $\wdtw(i,j) \le \awdtw(i,j) \le (1+\eps)
    \wdtw(i,j)$.

  \item Return $\awdtw(m,n)$.
\end{enumerate}

\begin{figure}
  \centering
  \begin{tabular}{ccc}
    \includegraphics[height=3cm]{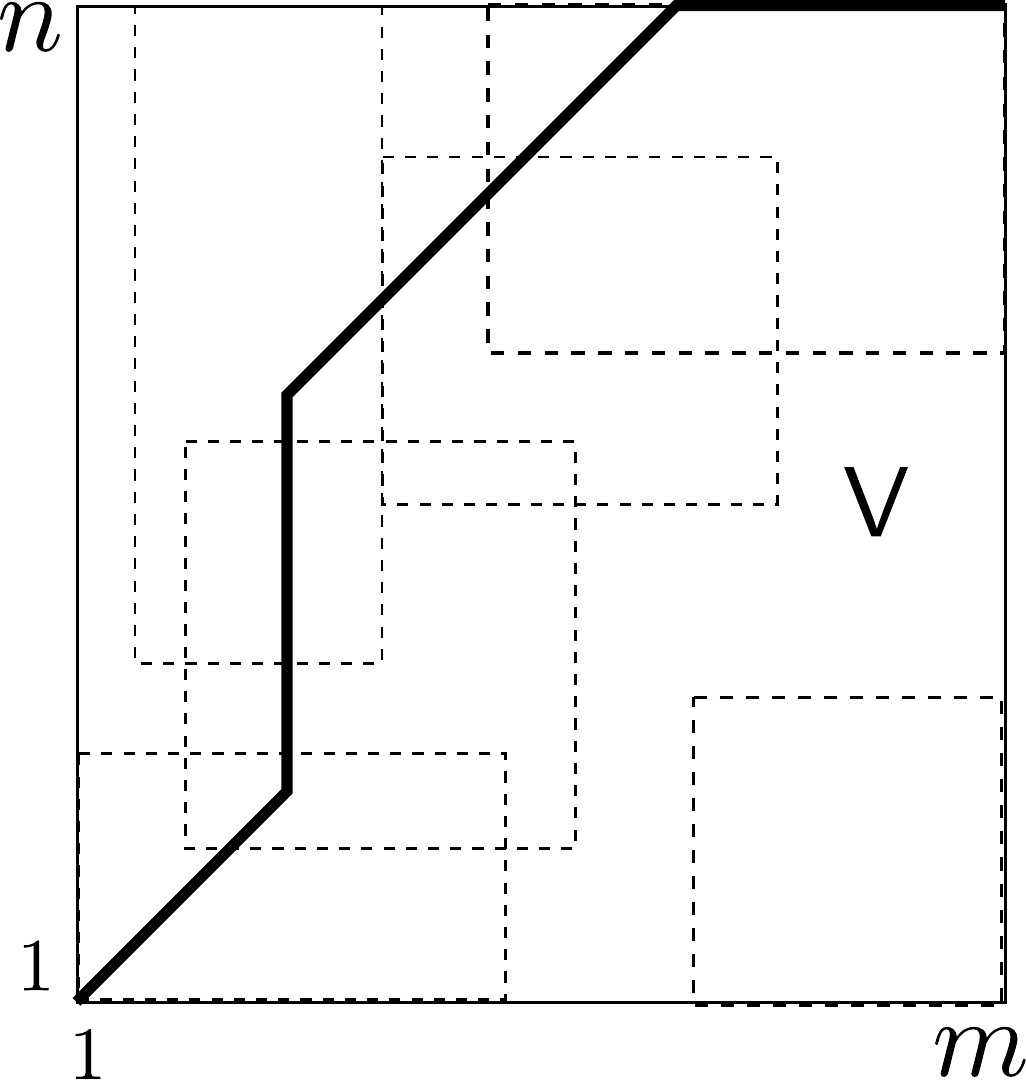} & \hspace{1cm} & 
    \includegraphics[height=3cm]{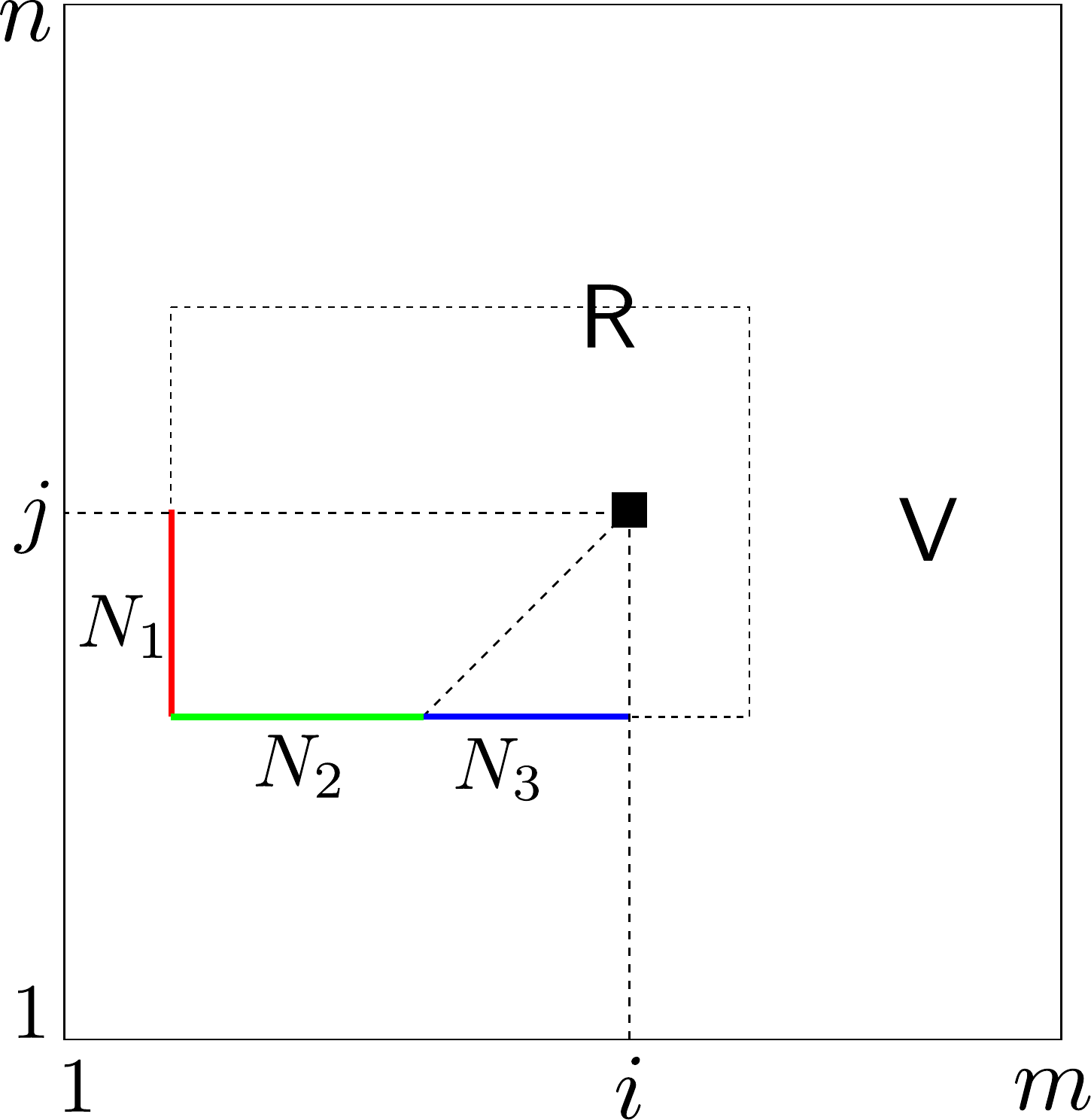} \\
    (a) & \hspace{1cm} & (b)
  \end{tabular}
  \caption{(a) Illustration of our algorithm: compute a small set of rectangles that covers
  the optimal admissible path from $(1,1)$ to $(m,n)$ (drawn in bold); (b) partitioning 
  the boundary points of rectangle $\rect$ dominated by $(i,j)$.}
  \label{fig:rectangles}
\end{figure}

\subsection{An $O(n)$ approximation} \label{sec:alg_step1}

Let $\dfr(P,Q)$ denote the discrete \frechet distance between $P$ and $Q$, i.e.,
replace sum with max in the definition of~$\dtw(P,Q)$.
\begin{lemma}
  $\displaystyle \dfr(P,Q) \le \dtw(P,Q) \le 2n \cdot \dfr(P,Q)$
  \label{lemma:n_approx}
\end{lemma}
\begin{proof}
  Let $\pi^*$ be the minimum-weight admissible path from $(1,1)$ to $(m,n)$ 
  corresponding to $\dtw(P,Q)$.  Then
  $$ \dfr(P,Q) \le \max_{(i,j) \in \pi^*} ||p_iq_j|| \le \sum_{(i,j) \in \pi^*} 
  ||p_iq_j|| = \dtw(P,Q).$$
  Similarly, let $\overline{\pi}$ be the admissible path corresponding to
  $\dfr(P,Q)$.  Then
  $$ \dtw(P,Q) \le \sum_{(i,j) \in \overline{\pi}} ||p_iq_j|| \le 2n \max_{(i,j) \in
    \overline{\pi}} ||p_iq_j|| = 2n \dfr(P,Q).$$
    The second inequality follows because $|\overline{\pi} \cap \V| \le m+n \le 2n$.
\end{proof}

Aronov \etal~\cite{AHKWW2006} gave a near-linear time algorithm for computing the
approximate discrete \frechet distance between $\kappa$-bounded point sequences.  
Their algorithm directly implies an $O(\kappa^d n\log
n)$-time 2-approximation algorithm for computing $\dfr(P,Q)$ for
$\kappa$-bounded sequences.  
They also prove that the same algorithm computes the
2-approximation of the discrete \frechet distance between backbone sequences in
$O(n^{2-2/d})$ time.  With a slight modification of their analysis, it turns out that 
their algorithm also works for
$\kappa$-packed sequences.  The details can be found in the Appendix.

\begin{lemma}[\hspace{-1sp}\cite{AHKWW2006}]
  A 2-approximation of $\dfr(P,Q)$ can be computed in time $O(\kappa n \log n)$,
  $O(\kappa^d n \log n)$, and $O(n^{2-2/d})$ if 
  $P,Q$ are $\kappa$-packed, $\kappa$-bounded, and backbone sequences, respectively.
  \label{lemma:2_approx}
\end{lemma}

Let $\overline{\dfr}(P,Q)$ be the 2-approximate discrete \frechet distance computed using
Lemma~\ref{lemma:2_approx}, i.e., $\dfr(P,Q) \le \overline{\dfr}(P,Q) \le 2\cdot \dfr(P,Q)$.  Set
$\dmin = \overline{\dfr}(P,Q)/2$. By Lemma~\ref{lemma:n_approx}, $\dmin \le \dtw(P,Q) \le 4n \dmin$.

\subsection{Computing rectangles $\rects$}
\label{sec:compute_rect_partition}

Let $\H$ be an axis-aligned
hypercube in $\R^d$ that contains $P\cup Q$.  Let $\T$ be a \emph{quadtree}, a $2^d$-way
tree, on $P\cup Q$.
Each node $v$ of~$\T$ is associated with an axis-aligned box $\sq_v$.  The root of
$\T$ is associated with $\H$.  A node $v$ is a leaf if $|\sq_v \cap (P\cup Q)| \le 1$.
The boxes associated with the children of a node $v$ are obtained by partitioning
$\sq_v$ into $2^d$ congruent hypercubes --- the side length of each resulting box is half
that of $\sq_v$.  
For a node $v \in \T$, let $\parent(v)$ denote its parent,
$\child(v)$ the set of children of $v$, $\Delta(v)$ the side length of $\sq_v$,
$P_v = P \cap \sq_v$, and $Q_v = Q \cap \sq_v$.  Let $m_v = |P_v|$ and
$n_v = |Q_v|$.  For two nodes $u,v \in \T$, let $\dsq(u,v) = \min_{p,q \in \sq_u\times
\sq_v} ||pq||$ denote the distance between $\sq_u$ and $\sq_v$.  
Two nodes $u, v$ are said to be \emph{neighboring} if $u$ and $v$ are in the same level
of $\T$ and $\sq_u$ and $\sq_v$ share a facet.
We do not construct the entire $\T$ but only a portion as described below.

Without loss of generality, suppose the side length of $\H$ is a power of 2.
Let $\rmin$ and $\rmax$ be powers of 2 such that $\rmin \le \eps
\dmin / 2n \le 2 \rmin$ and $\rmax \le 4\dmax \le 2\rmax$.
We call a node $v$ of $\T$ \emph{active} if $\Delta(v) \in [\rmin, \rmax]$ and $m_v+n_v >
0$. Let $\Aset$ denote the
set of active nodes of $\T$. We construct the set $\Aset$ of active nodes of $\T$ and the sets $P_v, Q_v$ for each
active node $v \in \Aset$.  By definition, the active nodes lie in a portion of the quadtree $\T$ of height
$O(\log(\rmax/\rmin)) = O(\log n)$.  Thus, $|\Aset| = O(n)$ and $\sum_{v \in \Aset} (m_v +
n_v) = O(n \log n)$.
Computing $\Aset$ and $P_v, Q_v$ for all $v \in \Aset$ takes $O(n \log n)$ time.

To compute the rectangles in $\rects$, we first construct a family
$\family = \Set{ (u_1,v_1), \ldots, (u_s, v_s)}$ of ``well-separated'' pairs of active nodes with the
following properties:
\begin{enumerate}[(P1)]
  \item For every $t \le s$, $\max\Set{ \Delta(u_t), \Delta(v_t)} \le \frac{\eps}{16} \max\Set{
        \dsq(u_t, v_t), \dmin/2n}$.

  \item For all pairs $(i,j) \in \V$ with $||p_i q_j|| \le
      \dmax$, there exists a unique pair $(u_t, v_t)$ such that $p_i \in P_{u_t}$ and $q_j \in
      Q_{v_t}$.
\end{enumerate}
Intuitively, $(u,v)$ is well-separated when for any $p \in u$ and $q \in v$, we have $||pq|| \approx \dsq(u, v)$.  
Then, for each pair $(u_t, v_t) \in \family$, we construct a small number of rectangles.

\paragraph*{Constructing $\family$.} 
The properties (P1) and (P2) are similar to those for the
so-called \emph{well-separated pair decomposition} (WSPD) of a point set, introduced by
Callahan and Kosaraju \cite{CK1995} (see also Har-Peled \cite{HP-book}).  We therefore adapt their
algorithm.  We first describe a recursive procedure \pair$(u,v)$, where $u,v$ are two
active nodes, which generates a family of pairs for $P_u, Q_v$.
\begin{center}
    \begin{tabular}{l}
        \hline
        \pair$(u,v)$ \\
        \textbf{if} $\max\Set{\Delta(u), \Delta(v)} \le \frac{\eps}{16} \max\Set{\dsq(u,v), \dmin/2n}$ \\
        ~ ~ add $(u,v)$ to $\family$; return \\
        \textbf{if} $\Delta(u) \ge \Delta(v)$, \textbf{then} \\
        ~ ~ $\forall w \in \child(u)$ ~ ~ if $P_w \neq \varnothing$, do \pair$(w,v)$ \\
        ~ ~ \textbf{else} $\forall z \in \child(v)$ ~ ~ if $Q_z \neq \varnothing$, do \pair$(u,z)$ \\
        \hline
    \end{tabular}
\end{center}
Let $u_0$ be a top-level active node with $\Delta(u_0) = \rmax$ and $P_{u_0} \neq
\varnothing$.  We call \pair$(u_0, v_0)$ if $Q_{v_0} \neq \varnothing$ and either
$v_0 = u_0$ or $v_0$ is a neighboring node of $u_0$.

(P1) is obviously true by the termination condition of the \pair~procedure.  
(P2) is true because for each $(i,j) \in \V$ with $||p_i q_j|| \le \dmax$, it must be that 
$p_i$ and $q_j$ are contained in either the same active node or two neighboring active
nodes of side length $\rmax$.  The stopping criterion ensures that the \pair~procedure 
never visits a node $v$ with $\Delta(v) < \rmin$.

By adapting the analysis of the WSPD algorithm, the following lemma
can be proven (see the appendix for the proof).

\begin{lemma}
    If $(u,v) \in \family$, then
        (i) $\max\Set{\Delta(u), \Delta(v)} \le \min\Set{ \Delta(\parent(u)),
            \Delta(\parent(v))}$; (ii)
        $\Delta(u)/2 \le \Delta(v) \le 2\Delta(u)$; and
        (iii) there is a constant $c \ge 0$ such that $\dsq(u,v) \le \frac{c}{\eps}
            \Delta(u)$.
    \label{lemma:wspd}
\end{lemma}

\paragraph*{Constructing $\rects$.}  
We describe how to construct rectangles from each
well-separated pair $(u,v) \in \family$.  Let $\hsq$ denote the concentric box to
$\sq$ with twice the side length of $\sq$.  The algorithm starts from the first unvisited point of
$P_u$ and walks along $P$ until $P$ exits $\hsq_u$; it then repeats this walk
by jumping to the next point of $P_u$; this process stops when 
all points of $P_u$ have been visited.  Each walk corresponds to a maximal 
contiguous subsequence (MCS) of $P$ in $\hsq_u$ with the first point inside $\sq_u$. 
Let $\Sub_u(P) = \{[x_1^-:x_1^+], \ldots, [x_{\alpha_u}^-:x_{\alpha_u}^+]\}$
denote the MCSs as constructed above.  
Similarly, we compute $\Sub_v(Q) = \{[y_1^-:y_1^+], \ldots,
  [y_{\beta_v}^-:y_{\beta_v}^+]\}$ denoting 
the MCSs of $Q$ in $\hsq_v$.
For every pair $a \in [\alpha_u], b \in [\beta_v]$, we define the
rectangle $\rect_{ab} = [x_a^{-}:x_a^{+}] \times [y_b^{-}:y_b^{+}]$ and set its
weight $\wei_{\rect_{ab}} = \dsq(u,v)$.  Set $\rects_{uv} = \{\rect_{ab} \mid a \in
[\alpha_u], b \in [\beta_v]\}$ and $\rects = \bigcup_{(u,v) \in \family} \rects_{uv}$.
See Figure~\ref{fig:contig_subseq} for an illustration of the MCSs of $\Sub_u(P),
\Sub_v(Q)$, and the rectangles in $\rects_{uv}$. 

\begin{figure}
  \centering
  \includegraphics[width=0.8\textwidth]{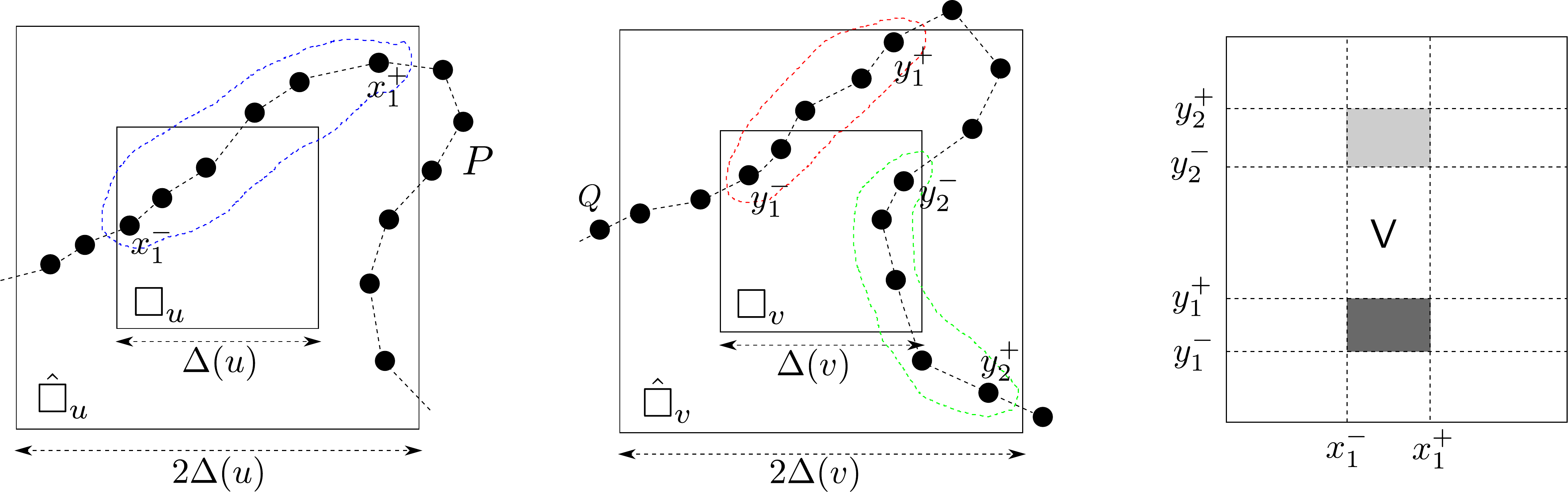}
  \caption{One MCS in $\hsq_u$ (left) and two MCSs in $\hsq_v$ (middle).  Together, two
    rectangles $\rect_{11}, \rect_{12}$ (right) are created (shaded areas).}
  \label{fig:contig_subseq}
\end{figure}

\emph{Remark.} The rectangles in $\rects_{uv}$ cover all the grid points corresponding to
$P_u \times Q_v$, i.e., if $(p_i, q_j) \in P_u \times Q_v$ then $(i,j) \in \bigcup
\rects_{uv}$.  Since $\bigcup \rects_{uv}$ may also contain grid points that correspond to
pairs in $(P \cap (\hsq_u \setminus \sq_u)) \times (Q \cap (\hsq_v \setminus \sq_v))$, a
grid point may lie in multiple rectangles of $\rects$, implying that the rectangles in
$\rects$ may overlap.  Had we defined $\Sub_u(P), \Sub_v(Q)$ to be MCSs of $P_u$ and
$Q_v$ respectively, the rectangles would have been disjoint, but we might have
ended up creating $\Omega(n^2)$ rectangles in the worst case.  As we will prove in
Section~\ref{sec:run_time}, by allowing $\rects_{uv}$ to cover extra points, we keep the
size of $\rects$ and $\boundaries$ small.

We show that the set of rectangles $\rects$ satisfies the conditions in step
(ii) of the algorithm.

\begin{lemma}
  $\rects$ satisfies the following properties:
  \begin{enumerate}[(i)]
    \item For all $\rect \in \rects$ and for all $(i,j) \in \rect$, 
      $|\wei(i,j) - \wei_{\rect}| \le \frac{\eps}{2} \max\{ \wei_{\rect}, \dmin/2n\}$.

    \item If $(i,j) \in \V$ and $\wei(i,j) \le \dmax$, then there exists a rectangle $\rect \in
      \rects$ such that $(i,j) \in \rect$. 
  \end{enumerate}
  \label{lemma:rectangle_property}
\end{lemma}
\begin{proof}
  (i) Suppose $\rect$ is constructed from some well-separated pair $(u_t, v_t) \in \family$.
  By construction, if $(i,j) \in \rect$, then $p_i \in \hsq_{u_t}$ and $q_j \in
  \hsq_{v_t}$.  Therefore, $\wei(i,j) = ||p_i q_j|| \le \dsq(u_t, v_t) + 2\sqrt{2}
  (\Delta(u_t) + \Delta(v_t))$.  By property (P1) and $\wei_{\rect} = \dsq(u_t, v_t)$, 
  we have $\wei(i,j) \le \wei_{\rect} + 4\sqrt{2} \frac{\eps}{16} \max \{\wei_{\rect},
\dmin/2n \} \le \wei_{\rect} + \frac{\eps}{2} \max \{\wei_{\rect}, \dmin/2n \}$.
Similarly, we can prove $\wei(i,j) \ge \wei_{\rect} - \frac{\eps}{2} \max\{\wei_{\rect}, \dmin/2n
  \}$.  (ii) By property (P2), there must exist a pair $(u, v) \in \family$
  such that $p_i \in P_{u}$ and $q_j \in Q_{v}$.  Since $\bigcup \rects_{uv}$ ``covers''
  $P_u \times Q_v$, there is a rectangle $\rect \in \rects_{uv}$ that contains the grid
  point $(i,j)$.
\end{proof}

The time taken to construct the set $\rects$ is $O(|\rects|)$ plus the time taken to
generate $\family$.  We bound the latter quantity in Section~\ref{sec:run_time}.

\subsection{Computing admissible paths} \label{sec:fill_table}

We now describe an algorithm that for each $(i,j) \in \boundaries$ computes
$(1+\eps)$-approximate value $\awdtw(i,j)$ of $\wdtw(i,j)$ in amortized constant time.

We say a point $(i,j) \in \V$ \emph{hits} a rectangle $\rect = [i_1:i_2] \times [j_1:j_2]$ 
if $i_1 < i \le i_2$ and $j_1 < j \le j_2$, i.e., $(i,j) \in \rect$ but not on its left or
bottom boundary. 
The algorithm sets~$\awdtw(1,1) = \wei(1,1)$,
and processes the points of $\boundaries$ from bottom to top and from left to
right in a row.  Suppose the current point is $(i,j)$.  There are two cases:
\begin{enumerate}[(i)]
  \item If $(i,j)$ does not hit any rectangle in $\rects$, we set
    \begin{equation}
      \awdtw(i,j) = \min\{ \awdtw(i-1, j), \awdtw(i, j-1), \awdtw(i-1, j-1) \} + \wei(i,j),
      \label{eq:compute_dtw1}
    \end{equation}
    where $\awdtw(a, b) = \infty$ if $(a,b) \notin \boundaries$. 

  \item  Let $\rect =
    [x^-:x^+] \times [y^-:y^+]$ be a rectangle hit by $(i,j)$.  We compute $\awdtw(i,j)$ using
    the computed $\awdtw$ values of the points on the left and bottom boundaries of
    $\rect$ that are dominated by $(i,j)$, and setting the weight of all points in
    $\rect$ to be $\wei_{\rect}$.  So the weight of a path inside $\rect$ is proportional to its length.  

    Without loss of generality, assume $i-x^- \ge j - y^-$. We divide the left and bottom
    boundary points that are dominated by $(i,j)$ into three sets: 
    $N_1 = \{(x^-, y) \mid y \in [y^-:j] \}$, $N_2 = \{(x, y^-) \mid x \in
    [x^-:i-(j-y^-)-1] \}$, and $N_3 = \{(x, y^-) \mid x \in [i-(j-y^-):i] \}$.  See
    Figure~\ref{fig:rectangles}(b).

    The optimal admissible path from $(1,1)$ to $(i,j)$ must pass through a
    point in $N_1 \cup N_2 \cup N_3$.  
    So we compute $\awdtw(i,j)$ as follows:
    \begin{equation}
      \awdtw(i,j) =  \min \left\{ 
        \begin{array}{ll} 
          (i-x^-)\wei_{\rect} &+ \min_{(a,b) \in N_1} \awdtw(a, b), \\ 
          i \wei_{\rect} &+ \min_{(a,b) \in N_2} (\awdtw(a,b) - a \wei_{\rect}), \\
          (j-y^-) \wei_{\rect} &+ \min_{(a,b) \in N_3} \awdtw(a,b) 
        \end{array} \right\}.
      \label{eq:compute_dtw2}
    \end{equation}
\end{enumerate}

The following lemma proves that our algorithm returns a $(1+\eps)$-approximation of $\dtw(P,Q)$.  
\begin{lemma}
  For each $(i,j) \in \boundaries$, if $\wdtw(i,j) \le \dmax$, then 
  $$|\awdtw(i,j) - \wdtw(i,j)| \le \frac{\eps}{2} (\wdtw(i,j) + (i+j)\dmin/2n ).$$
  \label{lemma:correctness}
\end{lemma}
\begin{proof}
  By induction on the order in which the $\awdtw$ values of points in $\boundaries$ are
  computed, we prove $\awdtw(i,j) - \wdtw(i,j) \le \frac{\eps}{2} (\wdtw(i,j) + (i+j)\dmin/2n
  )$.
  The lemma is obviously true for $(1,1)$.  Assume it is true for all points of
  $\boundaries$ processed before $(i,j)$.  We prove that it is also true for $(i,j)$.

  If $(i,j)$ does not hit any rectangle in $\rects$, then $\awdtw(i,j)$ is computed 
  using~(\ref{eq:compute_dtw1}).  Let $(a,b) \in \{ (i-1,j), (i, j-1), (i-1,j-1) \}$ be the
  predecessor of $(i,j)$ in the optimal admissible path from $(1,1)$ to $(i,j)$.
  Then $\wdtw(i,j) = \wdtw(a,b) + \wei(i,j)$.  Since $\wdtw(a,b) \le \wdtw(i,j) \le
  \dmax$, there is a rectangle $\rect$ containing $(a,b)$.
  Since~$(i,j)$ does not hit any rectangle, $(a,b)$ must actually lie on the boundary
  of~$\rect$, and thus in $\boundaries$.  So by induction hypothesis, 
  $$\awdtw(i,j) - \wdtw(i,j) = \awdtw(a,b) - \wdtw(a,b) \le \frac{\eps}{2} (\wdtw(a,b)+
  (a+b)\dmin/2n) \le \frac{\eps}{2} (\wdtw(i,j)+ (i+j)\dmin/2n).$$

  In the second case, let $\rect \in \rects$ be the rectangle hit by $(i,j)$ and used to
  compute $\awdtw(i,j)$.  
  Let $(a,b)$ be the intersection of the optimal admissible path from $(1,1)$ to $(i,j)$
  and the boundary of $\rect$.  
  Then 
  \begin{align*}
    \awdtw(i,j) &\le \awdtw(a,b) + \max\{i-a, j-b\} \wei_{\rect} \\
    &\le \wdtw(a,b) + \frac{\eps}{2} (\wdtw(a,b)+ (a+b) \dmin/2n ) + \max\{i-a, j-b\} \wei_{\rect}\\
    & \le \wdtw(i,j) + \frac{\eps}{2} ( \wdtw(i,j)+ (i+j) \dmin/2n ).
  \end{align*}
  The last inequality satisfies because $\wei_{\rect} \le \wei(h,k) + \frac{\eps}{2} \max\{\wei(h,k),
  \dmin/2n\}$ for any $(h,k) \in \rect$ by Lemma~\ref{lemma:rectangle_property}.
  Similarly, we can prove that $\wdtw(i,j) - \awdtw(i,j) \le \frac{\eps}{2} (\wdtw(i,j)+ (i+j)\dmin/2n
  )$, and the lemma follows.
\end{proof}

\begin{corollary}
  $|\awdtw(m,n) - \dtw(P,Q)| \le \eps \dtw(P,Q)$.
\end{corollary}

We now describe how the algorithm for computing $\awdtw(i,j)$ is implemented efficiently.

\paragraph*{Sorting points in $\boundaries$.}
Using radix sort, we sort the points of $\boundaries$ in $(y, x)$ lexicographical
order, where $x$ and $y$ denote the first and second coordinates of points, so that they
are sorted in the order in which they are processed.  We also perform
the same radix sort for $(x,y)$ and $(y-x,x)$ lexicographical orderings.  For each
point in $\boundaries$, we add a pointer to the previous point in each of the three
sorted orders, namely, a pointer to the first point below, to the left of, and on the
lower-left diagonal of the current point.  These pointers are used to identify
the $\awdtw$ values required in~(\ref{eq:compute_dtw1}). 

\paragraph*{Finding a rectangle hit by a point.}
The algorithm also needs to determine whether there exists a rectangle of $\rects$ hit by $(i,j)$.  This can be
achieved by maintaining the rectangle with the right-most right boundary when we
traverse each row.  More precisely, when processing the point $(i,j) \in \boundaries$, we
maintain a rectangle $\rect_{curr}$ that is hit by $(i,j)$ and whose right boundary spans
the farthest; we denote the $x$-coordinate of the right boundary of $\rect_{curr}$ by $\xi_{curr}$.  If no
rectangle hits $(i,j)$, we set $\rect_{curr} = \mathrm{NULL}$.  We update $\rect_{curr},
\xi_{curr}$ while processing $(i,j)$ as follows: If $(i,j)$ is the left boundary point of
a rectangle $\rect$ with $\xi$ being the $x$-coordinate of its right boundary and if $\xi >
\xi_{curr}$, we set $\rect_{curr} = \rect$ and $\xi_{curr} = \xi$.  Otherwise, if
$(i,j)$ is the right boundary point of $\rect_{curr}$, i.e., $i = \xi_{curr}$, we set
$\rect_{curr} = \xi_{curr} = \mathrm{NULL}$.  The total time spent at $(i,j)$ is
$O(1)$.

\paragraph*{Range-min data structure.}
If~$(i,j)$ hits the rectangle $\rect_{curr}$, we compute the minimum $\awdtw$ value in the
interval $N_1 = [y^-:j]$ on the left boundary of $\rect_{curr}$ and in the intervals
$N_2 = [x^-:i-(j-y^-)-1], N_3 = [i-(j-y^-):i]$ on the bottom boundary of $\rect_{curr}$, by performing range-min
queries.  We use the
data structure proposed by Fischer and Heun \cite{FH06} (see also \cite{RMQ}), which
answers a range-min query in $O(1)$ time after linear time preprocessing.  Thus, a range-min
query on $N_2$ or $N_3$ can be answered in $O(1)$ time by constructing a static range-min data
structure on the points on the bottom boundary of $\rect_{curr}$ (all $\awdtw$ values for these
points have been computed before visiting any point that hit $\rect_{curr}$).  On the other hand,
to support range-min query on $N_1$, we need a range-min data structure on the left 
boundary points of $\rect_{curr}$ that also supports inserting new points at the end when the $\awdtw$ values of
the left boundary points are computed row by row.

We briefly describe their static data structure, and show how to extend it
to support insertion in amortized $O(1)$ time.  
The input is an array of $k$ real numbers.  We say a data structure
has time complexity $\langle p(k), q(k) \rangle$ if the preprocessing takes time 
$O(p(k))$ and each query takes time $O(q(k))$.  The static data
structure divides the array into blocks of size $b = \frac{1}{4}\log_2 k$.  For each block, we
construct a naive $\langle b^2, 1\rangle$-time data structure.  But Fischer and Heun show
that many blocks can share the same data structure, so we create only much fewer than
$k/b = O(k/\log k)$ copies of the data structure.  Next the algorithm 
computes the minimum for each block and uses an $\langle k \log k, 1 \rangle$-time
``exponential-range'' data structure over the block minimums.  We now describe each of the
two structures in more detail.  

A \emph{Cartesian tree} of a
length-$b$ array stores the minimum element of the array at the root and the left
(resp.~right) child of the root is a Cartesian tree on the elements to the left
(resp.~right) of the minimum element.  It can be built by a linear scan of the elements and pushing/popping
elements into/from a stack at most $2b$ times; these push/pop operations serve as a fingerprint of the
Cartesian tree.  Thus, the number of different Cartisian trees for a
length-$b$ array is bounded by $2^{2b} = 4^b$.  It turns out that all
arrays of length $b$ that have the same structured \emph{Cartesian tree} \cite{CART} can share the 
same $\langle b^2, 1 \rangle$-time data structure.  We thus build $4^b$ copies of the
$\langle b^2,1\rangle$ data structure as follows: We go through each of the $O(k/\log
k)$ blocks, and compute the fingerprint of the block in $O(b)$ time; if there is no 
data structure corresponding to the fingerprint, we build it in $O(b^2)$ time by computing
the minimums for all possible $O(b^2)$ ranges.  
The exponential-range data structure maintains the minimums of $O(\log k)$ ranges starting
at each index $i \in [k]$ of exponentially increasing sizes $1, 2, 2^2, \ldots, 2^{\log
k}$.  Then the minimum of a range $[i,j]$ can be obtained by taking the minimum of two
ranges $[i, i+2^{\alpha}-1]$ and $[j-2^{\alpha}+1, j]$, where $\alpha$ is the largest
integer such that $2^{\alpha} \le j-i+1$.  The total
preprocessing time is $O( (k/b) \log (k/b) + 4^b b^2) = O(k)$.  

To answer a range-min query, we compute
the blocks containing the two end points of the query range; the minimum of the whole
blocks in the range can be answered using the exponential-range data structure in
$O(1)$ time; the minimums of the two partial blocks can also be answered in $O(1)$ time using
the naive data structures associated with the two boundary blocks.  So each query takes
$O(1)$ time.

We now describe how to support inserting an element to the end of the array in
amortized constant time.  If the last block of the array contains less than $b$
elements, the exponential-range data structure remains the same, and we just need to
update the fingerprint of the last block.  We can encode the fingerprint information (a
sequence of pushes and pops) as a path from the root to an internal node in a full binary 
tree of depth $2b$, where a push corresponds
to branching to the left child and a pop corresponds to branching to the right child.  
At each node of the binary
tree, we store a pointer to the corresponding naive range-min data structure.  Recall that the
Cartesian tree is built by a linear scan of the elements; so inserting a new element just
means going one more level down the binary tree, which takes constant time.  On the other
hand, when the last block is already full, the newly inserted element starts a new block.
In this case, we also need to update the exponential-range data structure, which takes
$O(\log (k/b))$ time; but since this only happens every $O(b) = O(\log k)$ elements, the
amortized time per insersion is still constant.  Therefore, we can insert an element to
the data structure in amortized $O(1)$ time.

\begin{lemma}
  For all $(i,j) \in \boundaries$, $\awdtw(i,j)$ can be computed in a total time of
  $O(|\boundaries|)$.
  \label{lemma:run_time_step3}
\end{lemma}

\subsection{Running time analysis} \label{sec:run_time}
We now bound the size of $|\boundaries|$, which by Lemma~\ref{lemma:run_time_step3} will
bound the running time of step (iii) of the algorithm.  Similar argument will bound the
time spent in generating the set $\family$, which will bound the running time of step
(ii).

\begin{lemma}
  The number of points in $\boundaries$ is $O(\frac{\kappa}{\eps} n \log n)$,
  $O(\frac{\kappa^d}{\eps^d} n \log n)$, and $O(\frac{1}{\eps} n^{2-1/d} \log n)$ for
  $\kappa$-packed, $\kappa$-bounded and backbone sequences, respectively.
  \label{lemma:num_boundary_pts}
\end{lemma}
\begin{proof}
  For any well-separated pair of quadtree nodes $(u,v) \in \family$, let $\hm_u = |P\cap
  \hsq_u|, \hn_v = |Q \cap \hsq_v|$.  Recall that $\Aset$ denotes the set of active nodes of
  quadtree $\T$, and $\alpha_u$ (resp.~$\beta_v$) is the number of maximal contiguous
  subsequences of $P$ in $\hsq_u$ (resp.~$Q$ in $\hsq_v$) computed by our algorithm.  
  Since $\hsq_u$ is concentric with $\sq_u$ with twice the side length of $\sq_u$,
  $\sum_{u\in \Aset} \hm_u \le 2^d \sum_{u\in \Aset} m_u = O(m \log n)$.
  Let $N(u) = \{v \mid (u,v) \in \family\}$.
  The total number of rectangle boundary points is 
  \begin{equation}
    |\boundaries| \leq 2 \sum_{(u,v) \in \family} (\hm_u \beta_v + \alpha_u \hn_v) 
    = 2 \sum_{u \in \Aset} \hm_u \sum_{v \in N(u)} \beta_v + 2 \sum_{v \in
    \Aset}\hn_v \sum_{u \in N(v)} \alpha_u.
    \label{eq:num_boundary_pts}
  \end{equation}
  We will show next that for any $u \in \Aset$, $\sum_{v\in N(u)} \beta_v$ is bounded
  by $O(\kappa/\eps)$ for $\kappa$-packed sequences, $O(\kappa^d/\eps^d)$ for
  $\kappa$-bounded sequences, and $O(n^{1-1/d}/\eps)$ for backbone sequences.
  The first part of (\ref{eq:num_boundary_pts}) is then bounded by
  $O(\frac{\kappa}{\eps} n \log n)$, $O(\frac{\kappa^d}{\eps^d} n \log n)$, and
  $O(\frac{1}{\eps} n^{2-1/d} \log n)$ for $\kappa$-packed, $\kappa$-bounded, and
  backbone sequences.  Symmetrically, the second part of (\ref{eq:num_boundary_pts}) has the
  same bound, and the lemma follows.

  We now bound $\sum_{v\in N(u)} \beta_v$ for any $u \in \Aset$.
  By Lemma~\ref{lemma:wspd}, there exists a constant $c$ such that for any $v \in N(u)$, 
  $\hsq_v$ is contained in a ball $\ball$ concentric with $\sq_u$ of
  radius $\frac{c}{\eps} \Delta(u)$.  

  There are two types of maximal contiguous subsequence $[y_b^-:y_b^+]$ of $Q \cap
  \hsq_v$ computed by our algorithm: (i) $q_{y_b^+} = q_n$ is the last point of $Q$, and
  (ii) $q_{y_b^++1}$, the point of $Q$ after the last point in the MCS lies outside of
  $\hsq_v$.  The first type of MCS is bounded by 
  the number of $v$'s such that $\hsq_v$ contains the last point of $Q$, $q_n$.  
  Suppose node $u$ is at level $t$ of the quadtree.
  By Lemma~\ref{lemma:wspd}(ii), such $v$'s can be from levels $t-1, t, t+1$.  Moreover,
  at each level, $q_n$ can be in the $\hsq_v$ of at most $2^d$ $v$'s.  Thus, the number of maximal
  contiguous subsequences of the first type is at most $2^{d}\times 3 = O(1)$.  In the
  following, we bound the second type separately for each family of input sequence.

  \textbf{\textit{$\kappa$-packed sequences}}.  Since the MCS starts inside $\sq_v$ and
  leaves $\hsq_v$ before the next MCS of $\Sub_v(Q)$ starts, the length of $\gamma_Q$
  between $q_{y_b^-}$ and $q_{y_b^+}$ is at least $\Delta(v)/2$ (see
  Figure~\ref{fig:contig_subseq}).  
  Let $\hL_v$ be the length of $\gamma_Q \cap \hsq_v$.  Then 
  $$ \sum_{v \in N(u)} \beta_v \le O(1) + \sum_{v \in N(u)} \frac{\hL_v}{\Delta(v)/2} \le
  O(1) + \frac{4}{\Delta(u)} \sum_{v\in N(u)} \hL_v.$$
  The last inequality follows from Lemma~\ref{lemma:wspd}(ii).  Because of the following
  four conditions: the side length of
  $\hsq_u$ is twice that of $\sq_u$, the nodes in $N(u)$ belong to three levels of
  $\T$ (Lemma~\ref{lemma:wspd}(ii)), the cells of the nodes at the same level of
  $\T$ are disjoint, and $\hsq_v \subseteq \ball$ for all $v \in N(u)$, we can conclude
  that 
  $$ \sum_{v \in N(u)} \hL_v \le 3\cdot 2^d |\gamma_Q \cap \ball| \le \frac{3\cdot 2^d c
  \kappa}{\eps}.$$
  The last inequality follows because $P$ is $\kappa$-packed sequence.  Hence $\sum_{v \in
  N(u)} \beta_v = O(\kappa/\eps)$, as claimed.

  \textbf{\textit{$\kappa$-bounded sequences}}. 
  We first show that for any two MCSs $[y_1^-:y_1^+]$ and $[y_2^-:y_2^+]$, $||q_{y_1^-} q_{y_2^-}|| \ge
  \Delta(v)/(\kappa+2)$.  This is because between points $q_{y_1^-}$ and $q_{y_2^-}$, the
  curve $\gamma_Q$ goes from inside $\sq_v$ to outside $\hsq_v$ which spans
  distance at least $\Delta(v)/2$.  Let $q$ be the intersection of this portion of
  $\gamma_Q$ with the boundary of $\hsq_v$.  By $\kappa$-boundedness, $\gamma_Q(y_1^-:y_2^-) \subseteq
  \ball(q_{y_1^-}, \frac{\kappa}{2} ||q_{y_1^-} q_{y_2^-}||) \cup \ball(q_{y_2^-},
  \frac{\kappa}{2} ||q_{y_1^-} q_{y_2^-}||)$.  Therefore, $(1+\kappa/2) ||q_{y_1^-}
  q_{y_2^-}|| \ge ||q_{y_1^-} q|| \ge \Delta(v) /2$, and the claim follows.  By a packing
  argument, the number of MCSs in $\hsq_v$ is bounded by $O(\kappa^d)$.
  Finally, $|N(u)| = O(1/\eps^d)$ 
  by another packing argument in the ball $\ball$.  So the number of second-type MCSs
  in all $\hsq_v$'s for $v \in N(u)$ is $O(\kappa^d/ \eps^d)$.
	
  \textbf{\textit{Backbone sequences}}. By the
  property that two consecutive points on a backbone sequence have distance between
  $c_1$ and $c_2$, there must exist one point on any MCS
  in the shell along the boundary of $\sq_v$ with thickness $c_2$. 
  The volumn of the shell is $O(\Delta(v)^d -
  (\Delta(v) - c_2)^d) = O(\Delta(v)^{d-1})$.  Furthermore, any two points on 
  $Q$ are at least distance $1$ apart.  So the number of MCSs is bounded by $O(\Delta(v)^{d-1})$. 
  Since $|N(u)| = O(1/\eps^d)$, the number of MCSs in all $\hsq_v$'s for $v \in N(u)$
  is $O(\frac{\Delta(v)^{d-1}}{\eps^d})$. On the other
  hand, each second-type MCS consumes a portion of $\gamma_Q$ of
  length at least $\Delta(v)/2$; this means that the subsequence contains 
  $\Omega(\Delta(v)/c_2) = \Omega(\Delta(v))$ points of $Q$. Since there are a total
  of $n$ points in $Q$, the total number of MCSs in all $\hsq_v$'s with $(u, v) \in \family$ is also bounded by
  $O(\frac{n}{\Delta(v)})$.
  The worst case happens when
  $ \frac{\Delta(v)^{d-1}}{\eps^d} = (\frac{n}{\Delta(v)})$, which leads to
  $\Delta(v) = \eps n^{1/d} $.  So the total number of second-type MCSs
  is $O(\frac{1}{\eps} n^{1-1/d})$.
\end{proof}

To bound the running time for constructing the family $\family$ of well-separated
node pairs, we bound, for each active node $u\in \Aset$, the number of times
\pair$(u,v)$ is called for some $v \in \Aset$.  This can be bounded
using similar arguments as the proof of Lemma~\ref{lemma:num_boundary_pts}, and we
can show that the time for constructing $\family$ is 
$O(\frac{\kappa}{\eps} n \log n)$, $O(\frac{\kappa^d}{\eps^d} n \log n)$, and
$O(\frac{1}{\eps} n^{2-1/d})$ for $\kappa$-packed, $\kappa$-bounded, and backbone
sequences, respectively.  Combining this bound with
Lemmas~\ref{lemma:2_approx},~\ref{lemma:run_time_step3}, and \ref{lemma:num_boundary_pts},
we obtain Theorem~\ref{thm:dtw}.

\def\aed{\hat{\ed}}
\def\wed{\mu}
\def\awed{\tilde{\wed}}

\section{Edit Distance} \label{sec:ED}
We now show how our DTW algorithm can be extended to compute a $(1+\eps)$-approximate
value of $\ed(P,Q)$.  Define $\V = [m+1] \times [n+1]$.
For $i < m+1$ and $j < n+1$, we have~$\wei(i,j) = ||p_i q_j||$.
Otherwise~$\wei(i,j)$ is undefined (e.g., $\wei(m+1, \cdot)$ and $\wei(\cdot, n+1)$
undefined).
We add an edge between every pair of neighboring points in $\V$.  The weight of a horizontal or vertical
edge is set to $g$ and the weight of a diagonal edge $\langle (i,j), (i+1, j+1) \rangle$
is set to $\wei(i,j)$.  The weight of an admissible path $\pi$ is defined as the sum of
weights of the edges along $\pi$.  As earlier, we define $\wed(i,j)$ to be the minimum
weight of an admissible path from $(1,1)$ to $(i,j)$.  Then $\ed(P,Q) = \wed(m+1, n+1)$.

We compute an approximate value of $\ed(P,Q)$ using the same 4-step algorithm as for $\dtw(P,Q)$, with
a different implementation of each step.  Step (ii)
of the algorithm remains the same, except that we add all the points on the
$(m+1)$-st column and $(n+1)$-st row of $\V$ to $\boundaries$.
In the
following, we give a simple $O(n)$-approximation for step (i), and point out
modifications needed in step (iii) to compute a value $\awed(i,j)$ for every $(i,j) \in
\boundaries$, such that $\wed(i,j) \le \awed(i,j) \le (1+\eps) \awed(i,j)$.

\paragraph*{$O(n)$-approximation.}
If $m = n$ and the weight of the monotone path~$\pi = \langle(1,1),
(2,2),\dots,(n+1,n+1)\rangle$ is at most $g$,
then it is returned as the optimal path.
Otherwise, $\ed(P,Q) \ge g$, and we set~$\dmin = g$.  Since $\ed(P,Q)$ is no larger than
the weight of an all-gap admissible path from $(1,1)$ to $(m+1, n+1)$, we have $\ed(P,Q)
\le 2(m+n)g \le \dmax = 4n\dmin$.

\paragraph*{Computing admissible paths.}
We describe how to compute $\awed(i,j)$,
for all $(i,j) \in \boundaries$, in the same row by row order.  
The main difference from DTW is that, 
since ED allows gaps, it is possible for the optimal admissible path between
$(1,1)$ and $(i,j)$ to have rectilinear subpaths through grid points that are not covered
by any rectangle.  
As in Section~\ref{sec:fill_table}, we consider whether there there exists a rectangle hit by
$(i,j)$.  

First, assume there exists a rectangle $\rect = [x^-:x^+] \times [y^-:y^+] \in \rects$ hit
by $(i,j)$.  Similar to DTW, we divide the relevant points on the left and bottom
boundaries of $\rect$ into three disjoint subsets $N_1, N_2$ and $N_3$.  If $2g \le
\wei_{\rect}$, it is always preferable to take a rectilinear path inside
$\rect$.  Thus we can assume the admissible path to $(i,j)$ goes through
either $(x^-, j)$ or $(i, y^-)$, and we set $\awed(i,j) = \min\{ \awed(x^-,j) +
(i-x^-), \awed(i,y^-) + (j-y^-)\}$.
If $2g > \wei_{\rect}$, the minimum-weight path inside $\rect$ should take as
many diagonal steps as possible.  So we set
\begin{equation*}
  \awed(i,j) = \min\left\{ 
    \begin{array}{l}
      j \wei_{\rect} + (i-j)g + \min_{(a,b) \in N_1} (\awed(a,b) + (g-\wei_{\rect})b), \\
      j \wei_{\rect} + (i-j)g + \min_{(a,b) \in N_2} (\awed(a,b) + (g-\wei_{\rect})b - ag), \\
      i \wei_{\rect} + (j-i)g + \min_{(a,b) \in N_3} (\awed(a,b) + (g-\wei_{\rect})a)
    \end{array} 
  \right\}.
\end{equation*}
We use the same range-min data structure of Fischer and Heun to compute each $\awed(i,j)$
in amortized $O(1)$ time.  The key used for the data structure is $\awed(a,b) +
(g-\wei_{\rect})b$, $\awed(a,b) + (g-\wei_{\rect})b - ag$, and $\awed(a,b) +
(g-\wei_{\rect})a$ for $N_1, N_2$, and $N_3$, respectively.


Next, assume $(i,j)$ does not hit any rectangle in $\rects$.  If $\{(i-1,j), (i, j-1),
(i-1, j-1)\} \subset \boundaries$, and thus their $\awed$ values have been computed, it is
trivial to compute $\awed(i,j)$ in $O(1)$ time.  We now focus on the case where one of the
predecessors of $(i,j)$ is not in $\boundaries$.
Let $U = \bigcup_{\rect \in \rects} \rect$ denote the union of all rectangles in
$\rects$. A point $(h,k) \in U$ is on the boundary of $U$, denoted by $\partial U$, if
$(h,k)$ does not lie in the interior of any rectangle of $\rects$; so at
least one point of
$\{(h-1,k), (h+1, k), (h, k-1), (h, k+1) \}$ is not in $U$.  Consider any admissible path $\pi$ from
$(1,1)$ to $(i,j)$ whose total weight is at most $\dmax$.  
Let $(a,b)$ denote the last point of $\boundaries$ on $\pi$ before
reaching $(i,j)$.  
The subpath of $\pi$ between $(a,b)$ and $(i,j)$ must be outside 
$U$, and it can only contain gaps since the weight of any point
outside $U$ is greater than $\dmax$, except the first step out of $(a,b)$, which
costs $\wei(a,b)$ if it is diagonal.
Let $(i_0, j)$ (resp.~$(i, j_0)$) be the first point of $\boundaries$ to the left of
(resp.~below) $(i,j)$ on row $j$ (resp.~column $i$).
Let $\partial U_{ij} = [i-1]\times [j-1] \cap \partial U$ denote the points on $\partial
U$ that are lower-left of
$(i,j)$.   We set
\begin{equation}
  \awed(i,j) = \min\left\{
    \begin{array}{l} 
      \awed(i_0,j) + (i-i_0)g,\\
      \awed(i,j_0) + (j-j_0)g,\\
      \min_{(a,b) \in \partial U_{ij}} (\awed(a, b) + (i+j - a - b)g + \min(0, \wei(a,b)-2g)) 
    \end{array} \right\}.
  \label{eq:compute_ed1}
\end{equation}
To compute $\awed(i,j)$, we use a different and simpler range-min 
data structure for $\partial U$ with key $\awed(a,b) - (a+b)g + \min(0, \wei(a,b)-2g)$,
that supports the decrease-key operation and query in $O(\log n)$ time.  More specifically, we
maintain a minimum over points of $\partial U$ in each column as we traverse $\boundaries$
row by row.  Then we maintain the column minimums in a complete binary tree where each
leaf corresponds to a column and an internal node stores the minimum over the leaves of
the subtree rooted at this node.  Note that the column minimums are always non-increasing
while we perform the row by row computation.
When the minimum of a column corresponding to some leaf $v$ gets decreased, we update the 
minimum information stored at each node along the path from $v$ to the root of the binary
tree.  This takes $O(\log n)$ time.  The last term in (\ref{eq:compute_ed1}) can be
computed by querying the the complete binary tree with range $[i-1]$ in $O(\log n)$ time.

\begin{lemma}
  The number of points on $\partial U$ is $O(\frac{\kappa}{\eps}n)$, $O(\frac{\kappa^d}{\eps^d} n)$, and
  $O(\frac{1}{\eps^{1-1/(d+1)}} n^{2-1/(d+1)})$ for $\kappa$-packed, $\kappa$-bounded, and 
  backbone sequences, respectively.
  \label{lemma:ed_bound}
\end{lemma}

\begin{proof}
  We analyze the number of times we enter of $U$ as we (conceptually) traverse each row 
  $j$ of $\V$ from column 1 to column $m$.  Let
  $\ball = \ball(q_j, 4\sqrt{d}~\dmax)$ be a ball centered at $q_j$ with radius
  $4\sqrt{d}~\dmax$.  Let $(i_1, j)$ 
  be an entering point of $U$, i.e., $(i_1,j) \in U$ and $(i_1-1, j) \notin
  U$, and let $(i_2, j)$ be the first point to the right of $(i_1,j)$ that is
  not in $U$.  We
  claim that the polygonal curve of the subsequence $P(i_1:i_2)$ has length $\Omega(\eps \dmax)$ inside
  $\ball$.  Let $\rect$ denote a rectangle of which $(i_1, j)$ is a
  (left) boundary point, and let $(u,v)$ denote the quadtree node pair from
  which $\rect_1$ is constructed, where $p_{i_1} \in \hsq_u$ and $q_j \in \hsq_v$.  
  If $\Delta(u) \le \frac{\eps}{2c} \dmax$, where
  $c$ is the constant in Lemma~\ref{lemma:wspd}(iii), then by Lemma~\ref{lemma:wspd}(iii),
  $\dsq(u, v) \le \frac{c}{\eps}\Delta(u) \le \dmax/2$.  Note that $\Delta(v) \le
  2\Delta(u) = O(\eps\dmax)$.  So $||p_{i_1} q_j|| \le \dsq(u,v) + O(\eps \dmax) \le
  \frac{3}{4} \dmax$.  Moreover, since $(i_2, j) \notin
  U$, it must be that $||p_{i_2} q_j|| > \dmax$.  Therefore, the polygonal curve
  of $P(i_1:i_2)$ contains $\Omega(\dmax)$ length inside ball $\ball$.  On the other hand, suppose $\Delta(u) >
  \frac{\eps}{2c} \dmax$.  Then $\hsq_u$ must be contained in ball $\ball$ 
  because the family of well-separated pairs are constructed by invoking
  the \pair~procedure on neighboring nodes of size $\dmax$.  So the polygonal curve of 
  $P(i_1:i_2)$ has $\Omega(\eps \dmax)$ length inside $\ball$ 
  because $P(i_1:i_2)$ goes from inside $\sq_u$ to outside of $\hsq_u$.  Thus, the claim
  is true.

  Using the same argument as in Lemma~\ref{lemma:num_boundary_pts}, we can show that the
  number of such entering points of $U$ is 
  $O(\frac{\kappa}{\eps})$ for $\kappa$-packed sequences and $O(\frac{\kappa^d}{\eps^d})$
  for $\kappa$-bounded sequences.  
  For backbone sequences, the bound is slightly weaker as
  we can not guarantee that there is a point in $P(i_1:i_2)$ within distance $c_2$ of the
  boundary of $\ball$.  Instead, the number of points of $P$ inside $\ball$ is bounded by 
  $O(\dmax^d)$ since any two non-consecutive points of $P$ have distance at least 1 (by
  definition of backbone sequences); this value
  is also an upper bound on the number of entering points of $U$ on row
  $j$.  A second upper
  bound of $O(m/(\eps\dmax))$ can be obtained similarly as in
  Lemma~\ref{lemma:num_boundary_pts}.  The worst case happens when $\dmax =
  (\frac{m}{\eps})^{1/(d+1)}$, which gives a bound of $(\frac{m}{\eps})^{1-
    \frac{1}{d+1}}$. Summing over $n$ rows, the bounds follows.
\end{proof}

Note that both the number of updates and the number of queries in the binary tree are
bounded by $|\partial U|$, and each update or query takes $O(\log n)$ time.  Moreover, the
case when there exists a rectangle hit by the current point takes the same time as for 
DTW.  Therefore Theorem~\ref{thm:ed} follows.

\section{Conclusion}

In this paper, we presented $(1+\eps)$-approximate algorithms for computing
the dynamic time warping (DTW) and edit distance (ED) between a pair of point sequences.
The running time of our algorithms is near-linear when the input sequences are
$\kappa$-packed or $\kappa$-bounded, and subquadratic when the input sequences
are protein backbone sequences.  Our algorithms are the first near-linear or
subquadratic-time algorithms known for computing DTW and ED for ``well-behaved'' 
sequences.  One interesting open question is whether there exists a near-linear 
algorithm for computing DTW and ED for backbone sequences in $\R^2$.  Another interesting open
problem is to identify other dynamic-programming based geometric optimization problems that can be 
solved using our approach, i.e., by visiting a small number of entries of the dynamic programming table using
geometric properties of the input.

\paragraph*{Acknowledgement.} The authors thank Sariel Har-Peled for many helpful
comments.

\bibliographystyle{abbrv}
\bibliography{dtw}

\newpage
\appendix

\section{2-approximation of $\dfr(P,Q)$}

Given a sequence $Q = \langle q_1, \dots, q_n\rangle$ in $\R^d$, a
subsequence $\tQ =
\langle Q_{i_1}, \dots, Q_{i_k} \rangle$ is a \emph{left $\mu$-simplification} of
$Q$ if 
\begin{enumerate}[(i)]
	\item $1\le i_1 < i_2 < \dots < i_k\le n$
	\item for any integer $j \in (i_{t-1}, i_t)$, $d(q_j, q_{i_{t-1}}) \le \mu$.
	\item for any $t \in [2:k]$, we have $d(q_{i_{t-1}}, q_{i_t}) > \mu$.
\end{enumerate}

We solve the decision version of the problem, i.e., deciding whether $\dfr(P,Q) < \delta$.
we perform a left $\mu$-simplification
on both sequences $P$ and $Q$ to obtain $\tilde{P}$ and $\tilde{Q}$ with $\mu = \eps
\delta / 2$.  According to \cite{AHKWW2006}, it suffice to solve the decision problem on
$\tilde{P}$ and $\tilde{Q}$ to obtain a $(1+\eps)$-approximation.  The time taken to solve
the decision problem is proportional to the number of entries in the dynamic programming
table that have value at most $\delta$; we call them reachable entries. This is equal 
to the sum of the number of points of $\tilde{P}$ in the ball $\ball(q, \delta)$ for 
all $q \in \tilde{Q}$.

Let $\gamma$ denote the $\kappa$-packed polygonal curve that connects the points of $\tilde{P}$.
Consider another ball $\ball'$ with the same center as $\ball$, but with a
radius $\delta+\mu$.  
By property of $\mu$-simplification, the portion of $\gamma$ between consecutive 
points in the sequence is of length at least $\mu$.
Thus every point of the simplified sequence inside $\ball$
is the left end-point of a unique portion of $\gamma$ of length at
least $\mu$ inside ball $\ball'$. By $\kappa$-packedness, the total length of
$\gamma$ inside $\ball'$ is $O(\kappa (\delta+\mu))$.  With $\mu =
O(\eps \delta)$, the total number of points in $\ball$ is bounded by 
$O(\kappa (\delta + \mu) / \mu) = O(\kappa (1 + \eps) / \eps) = O(\kappa / \eps)$, 
for $\epsilon \le 1$. 
Using the approximation algorithm for spherical range query introduced in the same paper,
these reachable entries can be computed in $O(\kappa / \eps)$ time.
Since there are at most $n$ points in 
$\tilde{P}$, the decision problem takes $O(n \kappa / \eps)$ time.

Using the technique of well-separated pair decomposition, we can obtain a set of $O(n)$
distances such that the distance between any pair of points from $P$ and $Q$ is 
$(1+1/5)$-approximated by one of the distances in the set. Therefore, by performing
$O(\log n/\eps) = O(\log n)$ binary searches, and using the above decision problem, 
we can obtain a $(1+\eps)$-approximation of $\dfr(P,Q)$ in $O(\frac{\kappa}{\eps} n \log n
)$ time.  Setting $\eps$ to a small constant, we obtain a 2-approximation.

\section{Proof of Lemma~\ref{lemma:wspd}}

First, we prove (i).  Suppose it is not true, i.e., $\max\{\Delta(u), \Delta(v)
\} > \min \{\Delta(\parent(u)), \Delta(\parent(v))\}$.  Without loss of
generality, assume $\Delta(u) \ge \Delta(v)$; then $\Delta(u) >
\Delta(\parent(v))$.  This means that the algorithm should have called
the \pair~procedure between the children of $u$ and $\parent(v)$, which is a
contradiction to the fact that \pair$(u,v)$ is called.

Next, we prove (ii).  Suppose $\Delta(v) < \Delta(u)/2$.  So $\Delta(\parent(v))
= 2 \Delta(v) < \Delta(u)$.  This is again a contradiction because
the algorithm should have called the \pair~procedure between $\parent(v)$ and
the children of $u$.  Similarly $\Delta(v) > 2\Delta(u)$ also leads to a
contradiction.

Finally, we prove (iii). Without loss of generality, suppose \pair$(u,v)$ is
recursively called by \pair$(\parent(u), v)$.  Since the algorithm does
not terminate at pair $(\parent(u), v)$, it must be that $\max\{\Delta(\parent(u)),
\Delta(v)\} > \eps \dsq(\parent(u), v)$.  So
$$\dsq(u,v) \le \dsq(\parent(u),v) + O(\Delta(u)) \le \frac{1}{\eps} \max\{\Delta(\parent(u)),
\Delta(v)\} + O(\Delta(u)),$$ which is $O(\Delta(u)/\eps)$ by (ii).

\end{document}